\documentclass[conference]{IEEEtran}
\IEEEoverridecommandlockouts
\usepackage{amsmath, graphicx, hyperref, xcolor, subcaption, braket, amsthm, amssymb, nicefrac}
\usepackage{algorithm} 
\usepackage{algpseudocode}
\usepackage{flushend}
\usepackage[nameinlink, noabbrev]{cleveref}

\def\x{{\mathbf x}}

\makeatletter
\def\thmheadbrackets#1#2#3{%
  \thmname{#1}\thmnumber{\@ifnotempty{#1}{ }\@upn{#2}}%
  \thmnote{ {\the\thm@notefont(#3)}}}
\makeatother

\newtheoremstyle{brakets}
  {}
  {}
  {\itshape}
  {}
  {\bfseries}
  {.}
  { }
  {\thmheadbrackets{#1}{#2}{#3}}

\theoremstyle{brakets}
\newtheorem{theorem}{Theorem}
\newtheorem{lemma}{Lemma}

\newtheorem{assumption}{Assumption}
\newtheorem{definition}{Definition}
\newtheorem{remark}{Remark}

\newtheorem*{problem*}{Problem}

\newtheorem*{objective*}{Objective}
\newtheorem*{key*}{Key Observation}

\def\BibTeX{{\rm B\kern-.05em{\sc i\kern-.025em b}\kern-.08em
    T\kern-.1667em\lower.7ex\hbox{E}\kern-.125emX}}
\begin{document}

\newcommand{\Or}{\mathcal{O}}
\newcommand{\loss}{\mathcal{L}}
\newcommand{\U}{\mathbf{U}}
\newcommand{\V}{\mathbf{V}}
\newcommand{\W}{\mathcal{W}}
\newcommand{\Ham}{\mathbf{H}}
\newcommand{\Op}{\mathbf{O}}
\newcommand{\ran}{\mathbf{v}}
\newcommand{\rant}{\ran^{(t)}}
\newcommand{\dd}{D_\ran(\z)}
\newcommand{\ddt}{D^{(t)}_\ran(\z)}
\newcommand{\ddcap}{\hat{D}_\ran(\z)}
\newcommand{\ddtcap}{\hat{D}^{(t)}_\ran(\z)}
\newcommand{\z}{\boldsymbol{\theta}}
\newcommand{\zt}{\boldsymbol{\theta}^{(t)}}
\newcommand{\ztp}{\boldsymbol{\theta}^{(t+1)}}
\newcommand{\nf}{\hat{f}\big(\z\big)}
\newcommand{\nft}{\hat{f}\big(\zt\big)}
\newcommand{\nftp}{\hat{f}\big(\zt+\frac{\pi}{2}\ej\big)}
\newcommand{\nftm}{\hat{f}\big(\zt-\frac{\pi}{2}\ej\big)}
\newcommand{\gp}{g_{\ran}^{\mu}}
\newcommand{\f}{f\big(\z\big)}
\newcommand{\ft}{f\big(\zt\big)}
\newcommand{\ftp}{f\big(\ztp\big)}
\newcommand{\F}{F\left(\z\right)}
\newcommand{\Ft}{F\left(\zt\right)}
\newcommand{\Ftp}{F\left(\ztp\right)}
\newcommand{\fs}{f^*}
\newcommand{\df}{\f-\fs}
\newcommand{\dft}{\ft-\fs}
\newcommand{\dftp}{\ftp-\fs}
\newcommand{\edf}{\mathbb{E}\left[\df\right]}
\newcommand{\edft}{\mathbb{E}\left[\dft\right]}
\newcommand{\edftp}{\mathbb{E}\left[\dftp\right]}
\newcommand{\eft}{\mathbb{E}\left[\ft \right]}
\newcommand{\eftp}{\mathbb{E}\left[\ftp \right]}
\newcommand{\eF}{\mathbb{E}\left[\F \right]}
\newcommand{\eFt}{\mathbb{E}\left[\Ft \right]}
\newcommand{\eFtp}{\mathbb{E}\left[\Ftp \right]}
\newcommand{\cedf}{\mathbb{E}\left[\df \vert \mathcal{F}_t \right]}
\newcommand{\cedft}{\mathbb{E}\left[\dft \big\vert \mathcal{F}_t \right]}
\newcommand{\cedftp}{\mathbb{E}\left[\dftp \big\vert \mathcal{F}_t \right]}
\newcommand{\ceF}{\mathbb{E}\left[\F \vert \mathcal{F}_t \right]}
\newcommand{\ceFt}{\mathbb{E}\left[\Ft \vert \mathcal{F}_t \right]}
\newcommand{\ceFtp}{\mathbb{E}\left[\Ftp \vert \mathcal{F}_t \right]}
\newcommand{\gt}{\boldsymbol{g}^{(t)}}
\newcommand{\gf}{\nabla\f}
\newcommand{\gfj}{\nabla_j\f}
\newcommand{\ngf}{\hat{\nabla}\f}
\newcommand{\ngfj}{\hat{\nabla}_j\f}
\newcommand{\gft}{\nabla\ft}
\newcommand{\ngft}{\hat{\nabla}\ft}
\newcommand{\ngftj}{\hat{\nabla}_j\ft}
\newcommand{\gftp}{\nabla\ftp}
\newcommand{\lb}{\left(}
\newcommand{\rb}{\right)}
\newcommand{\at}{\alpha_t}
\newcommand{\egt}{\mathbb{E} \left[ \lvert\lvert \nabla \ft \rvert\rvert^2 \right]}
\newcommand{\br}{\mathbb{R}}
\newcommand{\brd}{\mathbb{R}^d}
\newcommand{\WRP}{\par\qquad\(\hookrightarrow\)\enspace}
\newcommand{\roto}{\texttt{\textbf{Rotosolve}} }
\newcommand{\ej}{\mathbf{e}_j}
\newcommand{\gradcap}{\hat{\nabla}}
\newcommand{\maxeig}{\lambda_{\text{max}}}
\newcommand{\mineig}{\lambda_{\text{min}}}
\newcommand{\water}{H$_\text{2}$O }
\newcommand{\nhjt}{\hat{h}_j(\zt)}

\title{One Coordinate at a Time: Convergence Guarantees for Rotosolve in Variational Quantum Algorithms}

\author{\IEEEauthorblockN{1\textsuperscript{st} Sayantan Pramanik}
\IEEEauthorblockA{\textit{TCS Research, TATA Consultancy Services} \\
\textit{Indian Institute of Science}\\
Bangalore, India \\
sayantan.pramanik@tcs.com}
\and
\IEEEauthorblockN{2\textsuperscript{nd} M Girish Chandra}
\IEEEauthorblockA{\textit{TCS Research, TATA Consultancy Services} \\
Bangalore, India \\
m.gchandra@tcs.com}
}

\maketitle

\begin{abstract}
In this paper, we resolve an open question in the field of optimization algorithms for training parametrized quantum circuits: Does the popular \roto algorithm converge? Until now, interpolation-based coordinate descent methods such as \roto have mostly been treated as heuristics, lacking any formal convergence guarantees. We rigorously analyze \texttt{\textbf{Rotosolve}}, and show that it converges to $\varepsilon$-stationary points if the optimization landscape is non-convex and smooth; and to $\varepsilon$-suboptimal points if the objective function additionally obeys the Polyak-Łojasiewicz (PL) condition. Further, we derive explicit worst-case rates of convergence in the finite quantum measurement regime. These rates are contrasted against those from a similar coordinate-based method: Randomized Coordinate Descent (RCD). Although in the worst case their rates are, prima facie, equivalent, we present arguments for a more nuanced comparison between the two. We highlight that \roto is hyperparameter-free, and implicitly uses first and second derivatives in its updates. Finally, we supplement our theoretical findings with numerical experiments from Quantum Machine Learning; and compare the performance of \roto against RCD, Stochastic Gradient Descent, Simultaneous Perturbation Stochastic Approximation, and Randomized Stochastic Gradient Free methods.
\end{abstract}


\section{Introduction}\label{sec:intro}
Variational Quantum Algorithms (VQAs) \cite{vqe, vqa} rely on classical optimization subroutines to tune the parameters in quantum circuits. Although their efficiency is heavily dictated by the underlying optimizer, the theoretical understanding of these optimizers remains limited. In practice, one typically employs generic gradient-based (such as Stochastic Gradient Descent (SGD) \cite{sgd_handbook, sgd}, Randomized Coordinate Descent (RCD) \cite{Nesterov_coordinate_descent, Wright_coordinate_descent, rcd}, and Adam \cite{adam}) or gradient-free optimization methods (like Randomized Stochastic Gradient-Free (RSGF) \cite{Nesterov, Ghadimi_Lan}, or Simultaneous Perturbation Stochastic Approximation (SPSA) \cite{spsa1, spsa2}) which do not exploit the specific structure of quantum circuits or expectation values.

The \roto \cite{calculusonparameterizedquantumcircuits, smo, rotosolve, jacobidiagonalizationandersonacceleration} algorithm distinguishes itself by leveraging this structure. It first notes that when all but one parameter in the circuit are fixed, the univariate objective function can be expressed as a sinusoidal function of the free parameter. \roto exploits this observation by performing exact minimization along \emph{one coordinate at a time} by estimating the coefficients of the trigonometric function using a small number of circuit evaluations. The requisite number of evaluations required depends on the number of frequencies of the corresponding generator of the parametrized gate \cite{icd}. As a result, it can iteratively and optimally update individual parameters \emph{without resorting to step-sizes}.

Despite its widespread use and empirical success, no rigorous convergence guarantees or rates are known to exist \cite{icd}. In fact, it was recently empirically observed in \cite{icd} that \roto enjoys fast convergence when a high number of \emph{shots} is used, and fails to converge in the low-shot regime. We address this gap by showing that under general assumptions of smoothness (and optionally, invexity \cite{pl}), \roto provably converges. Furthermore, we derive explicit worst-case rates of convergence, and theoretically explain its dependence on shot-noise.

Besides convergence, we also focus on another important metric which dictates efficiency of the methods: its \emph{scalability} \cite{amira}. We show that the total number of shots required to reach a solution scales at most quadratically with the number of parameters in the circuit.

This work takes a first step toward a rigorous understanding of optimization methods tailored to variational quantum algorithms. We hope that the results herein will stimulate further research in this direction.

The rest of this paper is organized as follows: we summarize our main contributions in Section \ref{sec:contri}, and carefully describe the problem-setting and assumptions in Section \ref{sec:set}. A detailed overview of \roto is provided in Section \ref{sec:roto}, followed by a rigorous discussion on its convergence in Section \ref{sec:conv} and experimental validation in Section \ref{sec:exp}. Finally, we end with our conclusions in Section \ref{sec:con}.

\section{Main Contributions}\label{sec:contri}
The principal contributions made through this paper are summarized as follows:
\begin{enumerate}
    \item First, we demonstrate that the objective function arising from variational quantum tasks obeys a \emph{coordinate-wise Polyak-Łojasiewicz (PL) condition} \cite{pl} almost everywhere on the optimization landscape [\textcolor{blue}{Lemma \ref{lemma:pl}}].
    \item Next, to the best of our knowledge, we provide the first proof of convergence of the \roto algorithm.
    \item Additionally, we derive explicit worst-case rates of convergence and scalability-results that apply for different properties of the objective function. Briefly,
        \begin{enumerate}
            \item We show that if the objective function is non-convex and smooth, and the quantum circuit in consideration has $d$ optimizable parameters, then \roto requires $\Or(\nicefrac{d}{\varepsilon^2})$ iterations to reach an $\varepsilon$\emph{-stationary point} [\textcolor{blue}{Theorem \ref{thm:roto_smooth}}].
            \item Additionally, if the objective function also \emph{locally} obeys the coordinate-wise PL condition, then the algorithm can be shown to reach $\varepsilon$\emph{-suboptimal solutions} in $\Or(d\ln(\nicefrac{1}{\varepsilon}))$ iterations [\textcolor{blue}{Theorem \ref{thm:roto_pl}}].
        \end{enumerate}
    The shot-complexity of \roto has also been alluded to, in Section \ref{sec:shots}.
    \item Although these rates may look similar to those obtained by a related algorithm, RCD \cite{Nesterov_coordinate_descent, Wright_coordinate_descent, rcd} [\textcolor{blue}{Theorems \ref{thm:rcd_smooth}}, and \textcolor{blue}{\ref{thm:rcd_pl}}], we present subtle arguments in favor of \texttt{\textbf{Rotosolve}}. Most notably, we highlight that in comparison to RCD, the former is \emph{hyperparameter-free}; and implicitly utilizes \emph{second derivatives} (in addition to gradients) to speed up convergence.
    \item While \roto has been studied empirically in the context of VQAs, its application to Quantum Machine Learning (QML) tasks appears to be relatively unexplored. We demonstrate a straightforward adaptation of \roto for objective functions with a finite-sum structure, which are ubiquitous in machine learning \cite{finite_sum}. We then benchmark its performance against that of other optimizers such as SGD, RCD, RSGF, and SPSA on a binary classification problem.
\end{enumerate}

\section{Problem-Setting and Assumptions}\label{sec:set}
Unlike standard optimization algorithms, \roto exploits specific structural properties of the objective function. This necessitates a careful formulation of the problem-setting and underlying assumptions. We therefore devote a significant portion of the paper to explicitly formalize the mathematical formulation of the problem setup.

The initial state of the variational quantum circuit $\ket{\iota}$ is acted upon by a sequence of fixed and parametrized unitaries $\V_j$ and $\U_j(\z_j)$, such that the final state is $\ket{\psi(\z)} = \prod_{j=1}^d \V_j \U_j(\z_j)\ket{\iota}$. Here $j \in [d]$ and $\z \in \brd$ is the vector of optimizable parameters. The parametrized unitaries $\U_j(\z_j) = e^{-i\z_j\boldsymbol{\sigma}_j/2}$ are generated by the single-qubit Pauli matrices $\boldsymbol{\sigma}_j$, and the $\z_j$ are uncorrelated with each other. Finally, the expectation value of the Hermitian observable $\Ham$ is measured, which serves as our objective function:

\begin{problem*}\label{prob}
    \begin{equation}\label{eq:obj}
    \min_{\z \in \brd} \left( \f:= \braket{\psi(\z)|\Ham|\psi(\z)} \right).
\end{equation}
\end{problem*}

\noindent For the rest of this paper, we assume the following:

\begin{assumption}[Spectrum of the observable $\Ham$]
    The eigenvalues of the observable $\Ham$ lie in the interval $[-\mineig, \maxeig]$, where $\mineig$ and $\maxeig$ are real and non-negative.
\end{assumption}
\noindent This assumption ensures that the range of the objective function $f$ is also contained in $[-\mineig, \maxeig]$, i.e., $\f \in [-\mineig, \maxeig]$ for all $\z \in \brd$. For convenience, we use $\fs$ to denote the minimum value attained by $f$ over $\brd$.

\begin{assumption}[Smoothness of Objective Function \cite{rcd}]\label{as:smooth}
The objective function in Eq. \eqref{eq:obj} is coordinate-wise smooth. Formally, let $\ej$ represent the unit vector along any coordinate, with $j \in [d]$. Then, $\forall \, \z \in \brd$, and $\forall \, h \in \br$:
\begin{equation}
    |\nabla_j f(\z+\ej h)-\nabla_j\f| \leq L_j|h|.
\end{equation}
Here, $\nabla_j\f$ represents the $j^{\text{th}}$ component of the gradient vector with respect to $\z$; and $L_j$ is a positive, real number. Additionally, $f$ is $L$-smooth w.r.t. all of its coordinates considered together, i.e., $\forall \, \z, \boldsymbol{\phi} \in \brd$:
\begin{equation}
    ||\nabla\f-\nabla f(\boldsymbol{\phi})||_2 \leq L ||\z-\boldsymbol{\phi}||_2.
\end{equation}
\end{assumption}
\noindent It is easy to show that $L = \sum_j L_j$; and one may define the maximum and average smoothness constant over the coordinates $\bar{L} := \nicefrac{L}{d} \leq \max_j L_j =: L_{\text{max}}$ \cite{rcd}. Further, it can be shown that when the gate-generators are Pauli matrices, $L_j = \Or(\maxeig)$ for all $j \in [d]$, and $L = \Or(d\maxeig)$ \cite{liu_noise, vqa_smooth}. Henceforth, we will consider the smoothness assumption to always be satisfied for the rest of this paper.
 
\begin{assumption}[Local, Coordinate-wise PL condition]\label{as:pl}
Inspired by the local PL assumption in \cite{rcd}, we assume that the objective function in Eq. \eqref{eq:obj} locally obeys the coordinate-wise PL condition \cite{pl}. More specifically, let $\Theta$ denote the set of minimizers of $f$. Then, there exist $\delta, \mu_j > 0$ such that for any $\z \in \mathcal{N}(\Theta):= f^{-1}([\fs, \delta])$, the following holds for all $j \in [d]$:
\begin{equation}\label{eq:cwpl}
    \left(\nabla_j\f \right)^2 \geq 2\mu_j (\f-\fs) \geq 2\mu_j(\f - \fs_j).
\end{equation}
Here, $\fs_j$ is the coordinate-wise minimum of $f$, keeping all coordinates, but $j$, fixed. Moreover, 
\begin{equation}\label{eq:pl}
    ||\gf||^2_2 \geq 2 \mu (\f-\fs), 
\end{equation}
where, $\mu > 0$ is the PL-constant w.r.t. all the coordinates.
\end{assumption}
As with the smoothness condition above, one may easily verify that $\mu = \sum_j \mu_j$. Further, we show that the PL-condition is a valid assumption for variational quantum objective functions in Lemma \ref{lemma:pl}.

\begin{objective*}
    Under Assumption \ref{as:smooth} the goal of an optimization algorithm is to reach an $\varepsilon$-stationary solution of the objective function. Additionally, if Assumption \ref{as:pl} is satisfied, then we aim to obtain an $\varepsilon$-suboptimal point.    
\end{objective*}
\noindent  We formally define $\varepsilon$-stationarity and $\varepsilon$-suboptimality below:

\begin{definition}[$\varepsilon$-stationarity]
    A point $\z \in \brd$ is called an $\varepsilon$-stationary solution in the mean-squared sense if $\mathbb{E}[||\nabla \f||^2] \leq \varepsilon^2$.
\end{definition}

\begin{definition}[$\varepsilon$-suboptimality]
    A point $\z \in \brd$ is called an $\varepsilon$-optimal (or suboptimal) solution in the expected sense if $\mathbb{E}[\f-\fs] \leq \varepsilon$.
\end{definition}

Having stated all the assumptions on the nature of the objective function, we now discuss a more practical scenario. While the objective function has been defined exactly, in reality, one only has access to stochastic estimates $\nf$ of $\f$. We make the following assumption on the properties of such estimates.
\begin{assumption}[Properties of Stochastic Oracle]\label{as:prop}
    The finite-shot estimates of the objective function are unbiased, and have a bounded variance: $\mathbb{E}[\nf] = \f$, and $\mathbb{E}[(\nf-\f)^2] \leq \sigma^2$, for all $\z \in \brd$.
\end{assumption}

\begin{remark}
    It is easy to see that as a direct consequence of Assumption \ref{as:prop}, and the Parameter-Shift Rules \cite{psr, gpsr, ho_psr}, estimators of the coordinate-wise derivatives of $f$ are also unbiased have the same variance as above. Which is to say that for all $j \in [d]$ and $\z \in \brd$, $\mathbb{E}[\ngfj] = \gfj$, and $\mathbb{E}[(\ngfj-\gfj)^2] \leq \sigma^2$.
\end{remark}

\section{Overview of \roto}\label{sec:roto}
In this section, we provide a detailed description of the \roto algorithm, which unlike generic optimizers, leverages knowledge about the structure of parametrized quantum circuits and expectation values. When restricted to a single coordinate, the resulting univariate objective function assumes a trigonometric dependence on the free parameter. This observation forms the basis of \roto method, and is formalized through the following lemma:
\begin{lemma}[Sinusoidal form of Univariate Objective Function]\label{lemma:sin}
    For some $j \in [d]$, consider a univariate form $f_j:\br\rightarrow[-\mineig, \maxeig]$ of the objective function in Eq. \eqref{eq:obj}, where all but the $j^{\text{th}}$ coordinate of $\z$ have been fixed. Then, under the setting described in Section \ref{sec:set}:
    \begin{equation}\label{eq:sin}
        f_j(\phi) = A_j \sin(\phi+B_j)+C_j,
    \end{equation}
    where $A_j$, $B_j$, and $C_j$ are functions of $\z_{k\neq j}$. Additionally, if $\phi_\pm := \phi\pm\nicefrac{\pi}{2}$, $c:=\nicefrac{1}{2}(f_j(\phi_+)+f_j(\phi_-))$, $a:=2f_j(\phi)-2c$, and $b:= f_j(\phi_+)-f_j(\phi_-)$, then:
    \begin{align}
        &A_j = \frac{1}{2}\sqrt{a^2+b^2}, \label{eq:A} \\
        &B_j = \text{arctan2}(a, b)-\phi, \label{eq:B} \\
        &C_j = c \label{eq:C}.
    \end{align}
\end{lemma}
\begin{proof}
    We refer the reader to Appendix A in \cite{rotosolve} for a detailed proof of Lemma \ref{lemma:sin}.
\end{proof} 

\begin{remark}
    It was noted by Ding, et al, in \cite{rcd} that the objective functions arising from parametrized quantum circuits cannot obey the PL condition globally. The authors, thus, assumed the PL inequality to hold locally, within a neighborhood of the global minimizers. However, the validity of this assumption was not established. Equipped with the result in Lemma \ref{lemma:sin}, we verify that the local, coordinate-wise PL assumption \ref{as:pl} is indeed applicable in almost all of $\br$, and thus, $\brd$.
\end{remark}
\begin{lemma}\label{lemma:pl}
    Consider the univariate function in Eq. \eqref{eq:sin}, and let $\mathcal{X}^*:=\arg\max_{\phi \in \br}f_j(\phi)$. Under the setting of Section \ref{sec:set}, the local, coordinate-wise PL condition defined in Assumption \ref{as:pl} is satisfied at all points in $\br \setminus \mathcal{X}^*$, for all coordinates $j \in [d]$.
\end{lemma}
\begin{proof}
    The proof is straightforward. We note that $\fs_j = -A_j+C_j$, and plug $f_j$ and its derivative into Eq. \eqref{eq:cwpl} to obtain:
    \begin{equation}
        \mu_j \leq \frac{A_j \cos^2(\phi+B_j)}{2(1+\sin(\phi+B_j))}.
    \end{equation}
    Now, since $\mu_j > 0$, it may easily be verified that it is possible to get a positive value of $\mu_j$ at all points except in $\mathcal{X}^*$.
\end{proof}

We are now ready to describe the \roto algorithm. At each iteration, \roto selects a coordinate $j$ and freezes all other parameters in the quantum circuit. It initializes the free coordinate to a random value \cite{rotosolve}, say $\phi$, exploits the trigonometric nature of the resulting univariate objective function $f_j(\phi)$, and performs exact minimization along this coordinate. To do this, it expends a few circuit executions to compute the constant coefficients of the sinusoidal expression in Eq. \eqref{eq:sin}. In our current setup, where the generators of the parameterized gates have only a single frequency component, the quantum circuit may be evaluated individually at three points: $\phi$, $\phi+\nicefrac{\pi}{2}$, and $\phi-\nicefrac{\pi}{2}$, from which the coefficients may be computed using equations \eqref{eq:A}, \eqref{eq:B}, and \eqref{eq:C}. The coordinate is then updated to the minimizer of this one-dimensional function, i.e., $\phi^* = \arg \min_{\phi \in \br} f_j(\phi)$.

In a more general setting, \roto has been extended to work with quantum circuits where the generators have multiple frequency components \cite{excitation_solve}. If the generators have $r$ frequencies, then $2r+1$ evaluations of the quantum circuit are necessary to determine the coefficients of the trigonometric function \cite{icd}. Variants \cite{rotosolve, improvingvariationalquantumcircuit} of the algorithm have also been proposed where, besides optimizing the parameters, an optimal combination of parametrized gates is also chosen \cite{fraxis, fqs, fqs_controlled}. Orthogonally, \cite{rotolasso} proposes the inclusion of sparsity-inducing constraints; and \cite{icd} suggests computing the aforesaid coefficients by executing the quantum circuit at equidistant points to improve robustness to noise.

We now discuss a more practically-implementable version of the \roto algorithm, and make minor modifications which facilitate the proof of its convergence. First and foremost, one only has access to stochastic estimates of the objective function, as discussed in Section \ref{sec:set}. As a result, the coefficients $A$, $B$, and $C$ cannot be computed exactly. Rather, one performs \emph{inexact} coordinate minimization based on noisy estimates $\hat{A}$, $\hat{B}$, and $\hat{C}$ of the coefficients. Next, the parameters in \roto are typically updated cyclically or sequentially. However, it is difficult to establish convergence and derive explicit rates with such updates \cite{Nesterov_coordinate_descent}. Hence, at each iteration of the algorithm, we sample the free coordinate randomly from a uniform distribution.

Next, let $\zt$ denote the value of the parameters at the $t^{\text{th}}$ iteration. Following \cite{smo}, we set the free parameter to $\zt_j$, rather than initializing it randomly as in \cite{rotosolve}. While \cite{icd} recommends executing the quantum circuit at equidistant points around $\zt_j$, we instead use shifts of $\pm \nicefrac{\pi}{2}$. Although these choices appear to be inconsequential to the performance of the algorithm (albeit in the absence of noise), they are essential for convergence analysis in Section \ref{sec:conv}. We formally state the resulting algorithm in \ref{alg:roto}.

\begin{algorithm}
	\caption{\roto} 
    \label{alg:roto}
	\begin{algorithmic}[1]
            \State \textbf{Inputs:} Initial parameters $\z^{(0)}\in \mathbb{R}^d$, Number of iterations $T>0$.
		\For {$t \in [T]$}
                \State Select a coordinate $j \sim \text{Uniform}([d])$
                \State $\zt_+ \leftarrow \zt + \frac{\pi}{2}\ej$
                \State $\zt_- \leftarrow \zt - \frac{\pi}{2}\ej$
                \State Estimate $\nft, \hat{f}(\zt_+), \hat{f}(\zt_-)$ 
                \State $a \leftarrow 2\nft - \hat{f}(\zt_+) - \hat{f}(\zt_-)$
                \State $b \leftarrow \hat{f}(\zt_+) - \hat{f}(\zt_-)$
                \State $\hat{B} \leftarrow \text{arctan2}\left( a, b \right) - \zt_j$
			    \State $\ztp \leftarrow \zt - \left(\frac{\pi}{2} + \hat{B} \right)\ej $
		\EndFor
        \State \textbf{Output:} Sequence $\{\zt\}$
	\end{algorithmic} 
\end{algorithm}

Having discussed the \roto algorithm in detail, we now turn our attention to a related coordinate descent algorithm: RCD \cite{Nesterov_coordinate_descent, Wright_coordinate_descent, rcd}. This will help us better contextualize and contrast its convergence guarantees and performance. 

In a manner similar to Algorithm \ref{alg:roto}, RCD iteratively updates the decision variables one coordinate at a time. At any iteration $t$, it randomly chooses a coordinate direction $j$, computes the derivative along that direction, and takes a small step along the negative gradient. Mathematically, the update can be expressed as $\ztp = \zt - \alpha \ngftj \ej$, where $\alpha$ is a positive step-size,  and $\ej$ is the unit vector along the $j^{\text{th}}$ direction. Further, $\ngftj$ is a stochastic estimate of the noise-free derivative $\gfj$, which is computed using the Parameter-Shift Rules \cite{psr, gpsr, ho_psr}. In the current setting, evidently, each iteration of RCD requires the execution of two different quantum circuits. The algorithm has been stated in \ref{alg:rcd}.

\begin{algorithm}
	\caption{Randomized Coordinate Descent (RCD)} 
    \label{alg:rcd}
	\begin{algorithmic}[1]
            \State \textbf{Inputs:} Initial parameters $\z^{(0)}\in \mathbb{R}^d$, Number of iterations $T>0$, Sequence of step-sizes $\{\alpha_t\}$ with $t \in [T]$.
		\For {$t \in [T]$}
                \State Select a coordinate $j \sim \text{Uniform}([d])$
                \State Estimate $\gradcap_j \ft$
                \Comment{Using the Parameter-Shift Rules \cite{psr}}
                \State $\ztp \leftarrow \zt -\alpha_t \gradcap_j \ft \ej$
		\EndFor
        \State \textbf{Output:} Sequence $\{\zt\}$
	\end{algorithmic} 
\end{algorithm}

\begin{remark}
    It is noteworthy that, unlike RCD, \roto does not depend on step-sizes, and is completely free of any hyperparameters. While the former takes a small step to ensure a reduction in the objective value at each iteration, for the given values of the other frozen parameters, \roto finds the global minimum along the chosen coordinate. As a result, while RCD is known as a coordinate \emph{descent} method, \roto is a coordinate \emph{minimization} algorithm.
\end{remark}

\section{Discussion on Convergence}\label{sec:conv}
In this section, we prove the convergence of the \roto algorithm. We also derive explicit rates of its convergence to an $\varepsilon$-stationary solution under the smoothness assumption \ref{as:smooth}, and to an $\varepsilon$-suboptimal solution under the additional assumption that the PL condition is satisfied \ref{as:pl}. 

\subsection{Iteration Complexity of \roto}\label{sec:conv_roto}
\begin{key*}
    Before proceeding with the convergence results, we take a brief pause to present a key observation about the \roto algorithm. From \cite{ho_psr}, we note that the first and second derivatives of the objective function in Eq. \eqref{eq:obj} with respect to the $j^{\text{th}}$ coordinate, are given by:
    \begin{align}
        &\nabla_j\f = \frac{1}{2}\left(f\left(\z+\frac{\pi}{2}\ej\right)-f\left(\z-\frac{\pi}{2}\ej\right)\right), \;\; \text{and,} \\
        &h_j(\z) = \frac{1}{2}\left(f\left(\z+\frac{\pi}{2}\ej\right)+f\left(\z-\frac{\pi}{2}\ej\right)-2\f\right),
    \end{align}
    respectively. Plugging them into equations \eqref{eq:A} and \eqref{eq:C}, we get:
    \begin{align}
        &A_j = \sqrt{(h_j(\z))^2+(\nabla_j\f)^2}, \;\; \text{and,} \label{eq:AA} \\
        &C_j = h_j(\z)+\f. \label{eq:CC}
    \end{align}
    This shows that \roto implicitly utilizes first and second derivatives to perform coordinate-wise minimization.
\end{key*}

We are now ready to derive the rates of convergence for \texttt{\textbf{Rotosolve}}. We start with the following lemma which upper bounds the expected reduction in objective value from one iteration to the next.
\begin{lemma}[\roto Descent Lemma]\label{lem:roto_descent}
    Under the setting described in Section \ref{sec:set}, if Algorithm \ref{alg:roto} is executed for $T\geq 1$ iterations, then for every $t \in [T]$ the following holds true:
    \begin{equation}\label{eq:roto_lem}
        \mathbb{E}[\ftp-\ft ] \leq -\frac{\mathbb{E}[||\gft||^2]-\sigma^2}{(1+\sqrt{2})\maxeig}.
    \end{equation}
\end{lemma}
\begin{proof}
    At iteration $t$, let the random coordinate $j$ be selected. It is easy to note that \cite{rotosolve}:
    \begin{equation*}
        \ftp = \fs_j = -\hat{A}_j + \hat{C}_j.
    \end{equation*}
    Substituting the values of $\hat{A_j}$ and $\hat{C}_j$ from equations \eqref{eq:AA} and \eqref{eq:CC}, and subtracting $\ft$ from both sides, we have:
    \begin{equation}\label{eq:1}
    \begin{aligned}
        \ftp - \ft = & -\sqrt{(\nhjt)^2+(\ngftj)^2} \\ &+\nhjt + (\nft-\ft).
    \end{aligned}
    \end{equation}
    Now, we note that:
    \begin{equation*}
    \begin{aligned}
        -\sqrt{(\nhjt)^2+(\ngftj)^2} +\nhjt = \\ -\frac{(\ngftj)^2}{\sqrt{(\nhjt)^2+(\ngftj)^2} +\nhjt}
    \end{aligned}
    \end{equation*}
    The denominator of the r.h.s of the above equation is evidently upper bounded by $u:= (1+\sqrt{2})\maxeig$. Plugging this back into Eq. \eqref{eq:1},
    \begin{equation}
        \ftp - \ft = -\frac{(\ngftj)^2}{u} + (\nft-\ft).
    \end{equation}
    Next, we take an expectation conditioned on $\zt$, and $j$:
    \begin{equation*}
    \mathbb{E}[\ftp-\ft \; \lvert \; \zt, j ] \leq -\frac{(\nabla_j\ft)^2}{u} + \frac{\sigma^2}{u},
    \end{equation*}
    where, we have made use of Assumption \ref{as:prop}. Finally, taking an expectation w.r.t. $j$, followed by $\zt$, gives us the required result.
\end{proof}

Now, we consider the smooth, non-convex case, and find the convergence rate of \roto to an $\varepsilon$-stationary point:
\begin{theorem}[Convergence of \roto to an $\varepsilon$-stationary point]\label{thm:roto_smooth}
    Under Assumption \ref{as:prop}, if Algorithm \ref{alg:roto} is run for $T\geq 1$ iterations, then:
    \begin{equation}
        \min_{t \in [T]} \mathbb{E}[||\nabla \ft||^2] \leq \frac{(1+\sqrt{2})\maxeig d}{T}(f(\z^{(0)})-\fs) + d \sigma^2.
    \end{equation}
    Further, for any small $\varepsilon>0$, if $\sigma^2 \leq \nicefrac{\varepsilon^2}{2d}$, then it can be guaranteed that $\min_{t \in [T]} \mathbb{E}[||\nabla \ft||^2] \leq \varepsilon^2$ in $T$ iterations, where,
    \begin{equation}
        T \geq \frac{2(1+\sqrt{2})d\maxeig}{\varepsilon^2}\left(f(\z^{(0)})-\fs\right).
    \end{equation}
\end{theorem}
\begin{proof}[Proof sketch]
    We start with the result in Lemma \ref{lem:roto_descent} and unroll the recurrence relation in Eq. \ref{eq:roto_lem}. This leads to a telescopic cancellation of the $\ft$ terms. Next, we note that $\fs \leq f(\z^{(T)})$. Finally, we assert that the upper bound on $\min_{t \in [T]} \mathbb{E}[||\nabla \ft||^2]$ should be less than $\varepsilon^2$ to obtain the result.
\end{proof}

Finally, we find the iteration complexity of \roto to reach an $\varepsilon$-suboptimal solution:
\begin{theorem}[Convergence of \roto to an $\varepsilon$-suboptimal solution]\label{thm:roto_pl}
    In the Setting of Section \ref{sec:set}, let Assumptions \ref{as:pl} and \ref{as:prop} be satisfied. If Algorithm \ref{alg:roto} is run for $T\geq 1$ iterations, then for every $t \in [T]$:
    \begin{equation}
    \begin{aligned}
        \mathbb{E}[\ftp-\fs] \leq &\left(1-\frac{2\mu}{(1+\sqrt{2})d\maxeig} \right) \mathbb{E}[\ft-\fs] \\
        &+ \frac{\sigma^2}{(1+\sqrt{2})\maxeig}.
    \end{aligned}
    \end{equation}
    Further, for any small $\varepsilon > 0$, if $\sigma^2 \leq \nicefrac{\varepsilon\mu}{d}$, Algorithm \ref{alg:roto} is guaranteed to reach an $\varepsilon$-stationary solution in $T$ iterations, where,
    \begin{equation}
        T \geq \frac{(1+\sqrt{2})d\maxeig}{2\mu}\ln\left( \frac{2(f(\z^{(0)})-\fs)}{\varepsilon}\right).
    \end{equation}
\end{theorem}
\begin{proof}[Proof sketch]
    Similar to before, we start with the result in Lemma \ref{lem:roto_descent} and plug in the PL condition from Assumption \ref{as:pl}, followed the subtraction of $\fs$ from both sides of the resulting inequality. Next, we take a telescoping sum, and assert that the expected optimality gap at the final iteration should be less than $\varepsilon$.
\end{proof}

\subsection{Iteration Complexity of RCD}
For ease of comparison with \texttt{\textbf{Rotosolve}}, we briefly touch upon the iteration complexity of RCD. 
First, we state a coordinate-wise version of the Descent Lemma that is central to proving convergence for smooth functions.
\begin{lemma}[Coordinate-wise Descent Lemma]
    In the setting of Section \ref{sec:set}, let Assumption \ref{as:smooth} be satisfied. Then for all $\z \in \brd$, $h \in \br$, and $j \in [d]$, the following holds true:
    \begin{equation}
        f(\z+h\ej) \leq \f + h\braket{\gf, \ej} + \frac{L_jh^2}{2}.
    \end{equation}
\end{lemma}
\begin{proof}
    This is a corollary of Lemma 2.25 in \cite{sgd_handbook} along a single dimension.
\end{proof}

The following Lemma and Theorems are direct analogues of those in Section \ref{sec:conv_roto}, and their proofs proceed similarly.
\begin{lemma}[RCD Descent Lemma]\label{lem:roto_descent}
    Under the setting described in Section \ref{sec:set}, if Algorithm \ref{alg:roto} is executed with a step-size $\alpha>0$ for $T\geq 1$ iterations, then for every $t \in [T]$ the following holds true:
    \begin{equation}\label{eq:roto_lem}
        \mathbb{E}[\ftp-\ft ] \leq -\frac{\alpha}{2d}\mathbb{E}[||\gft||^2]+\frac{\bar{L}\alpha^2\sigma^2}{2}.
    \end{equation}
\end{lemma}

Now, we consider the smooth, non-convex case, and find the convergence rate of RCD to an $\varepsilon$-stationary point:
\begin{theorem}[Convergence of RCD to an $\varepsilon$-stationary point]\label{thm:rcd_smooth}
    Under Assumption \ref{as:prop}, if Algorithm \ref{alg:rcd} is executed for $T\geq 1$ iterations with a fixed step-size $\alpha \leq \nicefrac{1}{L_{\text{max}}}$, then:
    \begin{equation}
        \min_{t \in [T]} \mathbb{E}[||\nabla \ft||^2] \leq \nicefrac{2d}{\alpha T}(f(\z^{(0)})-\fs) + \bar{L} \alpha d \sigma^2.
    \end{equation}
    Consequently, it can be guaranteed that for any small $\varepsilon>0$, if $\sigma^2 \leq \nicefrac{\varepsilon^2}{2d}$, then $\min_{t \in [T]} \mathbb{E}[||\nabla \ft||^2] \leq \varepsilon^2$, and,
    \begin{equation}
        T \geq \frac{4dL_{\text{max}}}{\varepsilon^2}\left(f(\z^{(0)})-\fs\right).
    \end{equation}
\end{theorem}

Finally, we find the iteration complexity of RCD to reach an $\varepsilon$-suboptimal solution:
\begin{theorem}[Convergence of RCD to an $\varepsilon$-suboptimal solution]\label{thm:rcd_pl}
    In the Setting of Section \ref{sec:set}, and under Assumptions \ref{as:pl} and \ref{as:prop}, if Algorithm \ref{alg:rcd} is executed for $T\geq 1$ iterations with a fixed step-size $\alpha \leq \nicefrac{1}{L_{\text{max}}}$, then for every $t \in [T]$:
    \begin{equation}
        \mathbb{E}[\ftp-\fs] \leq \left(1-\frac{\alpha\mu}{d} \right) \mathbb{E}[\ft-\fs] + \frac{\bar{L}\alpha^2\sigma^2}{2}.
    \end{equation}
    Further, if $\sigma^2 \leq \nicefrac{\varepsilon\mu}{d}$, Algorithm \ref{alg:rcd} reaches an $\varepsilon$-suboptimal solution in $T$ iterations, where,
    \begin{equation}
        T \geq \frac{dL_{\text{max}}}{\mu} \ln\left(\frac{2(f(\z^{(0)})-\fs)}{\varepsilon} \right).
    \end{equation}
\end{theorem}

\subsection{Comparison of Convergence Rates}\label{sec:shots}
Let us now compare the convergence rates of \roto against those of RCD. First, we reiterate that $L_{\text{max}} = \Or(\maxeig)$. It is easy to note that if the conditions on noise-variance in Theorems \ref{thm:roto_smooth}-\ref{thm:rcd_pl} are satisfied, then \roto and RCD have equivalent iteration complexities. Further, if $n$ shots are used to estimate the objective values, then the variance is bounded by $\nicefrac{\maxeig^2}{n}$. This also makes the shot-complexity of the two methods equal, up to a small constant factor. However, these are worst-case results, and may not reflect practical performance, which we verify through the experiments in Section \ref{sec:exp}.

\section{Experiments and Results}\label{sec:exp}
In this section, we benchmark the performance of \roto against that of SGD, RCD, RSGF, and SPSA, on a binary classification task. While \roto has been tested extensively on applications in Chemistry and Combinatorial Optimization, to the best of our knowledge, the application of \roto to QML has not been explored in prior work. First, we extend \roto to be compatible with objective functions having a finite-sum structure \cite{finite_sum}, which are frequently encountered in machine learning. Let $\{(\x_k, y_k)_{k \in [K]}\}$ denote a training dataset, where $\x_k \in \br^m$ are the input features and $y_k \in \{\pm 1\}$ is the ground-truth corresponding to the $k^{\text{th}}$ data point. Let the value $f(\x_k, \z) \in [-1, 1]$ measured from a quantum circuit adhere to the assumptions in Section \ref{sec:set}. For simplicity, we consider the loss \cite{ssd}:
\begin{equation}
    \loss(\z) = \frac{1}{2K}\sum_{k=1}^K (1-y_k f(\x_k, \z)).
\end{equation}
Now, if all but the $j^{\text{th}}$ parameter in the ansatz are fixed, then the loss for each training example can be expressed as $\loss^{(k)}_j(\z_j) = a^{(k)}_j\sin(\z_j)+b^{(k)}_j\cos(\z_j)+c^{(k)}_j$. The coefficients $a_j$ and $b_j$ may be evaluated as:
\begin{align}
    &a_j = \sum_k (2\loss^{(k)}_j(\nicefrac{\pi}{2})-\loss^{(k)}_j(0)-\loss^{(k)}_j(\pi))/2K, \\ &b_j = \sum_k (\loss^{(k)}_j(0)-\loss^{(k)}_j(\pi))/2K. 
\end{align}
The value of the $j^{\text{th}}$ parameter is then updated to:
\begin{equation}
    \z_j^* = \text{arctan2}(a_j, b_j)+\pi
\end{equation}

Next, we replicate the experimental setup from \cite{ssd}. Specifically, we use the same encoding, variational ansatz, initialization, and hyperparameters for the various optimization algorithms (including SPSA and RSGF with a perturbation parameter of $10^{-2}$). To simulate the effect of a finite number of shots, Gaussian noise with zero mean, and $\sigma=10^{-4}$ has been added to the objective values \cite{icd}. The mean and standard deviation of the training losses across 10 independent runs (with the same initial point) against the number of iterations and circuit-executions has been depicted in Fig. \ref{fig:res}. While \roto is found to outperform all the baselines and reach lower loss values, it also has the highest standard deviation of all the algorithms. This reinforces the observations made in \cite{icd}.

\begin{figure}
    \centering
    \begin{subfigure}{\columnwidth}
        \centering
        \includegraphics[width=0.88\linewidth]{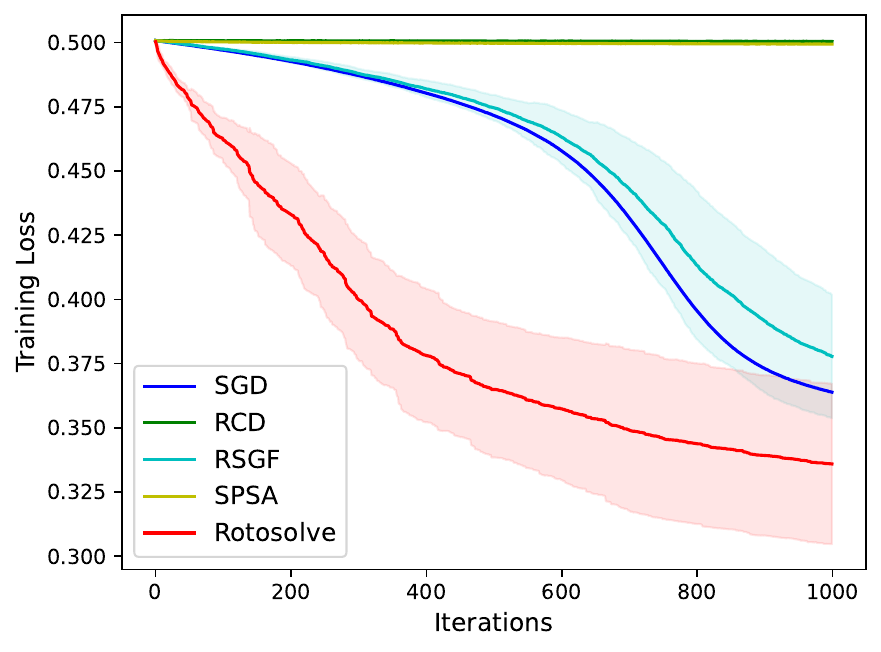}
        \caption{}
    \end{subfigure}
    \begin{subfigure}{\columnwidth}
        \centering
        \includegraphics[width=0.88\linewidth]{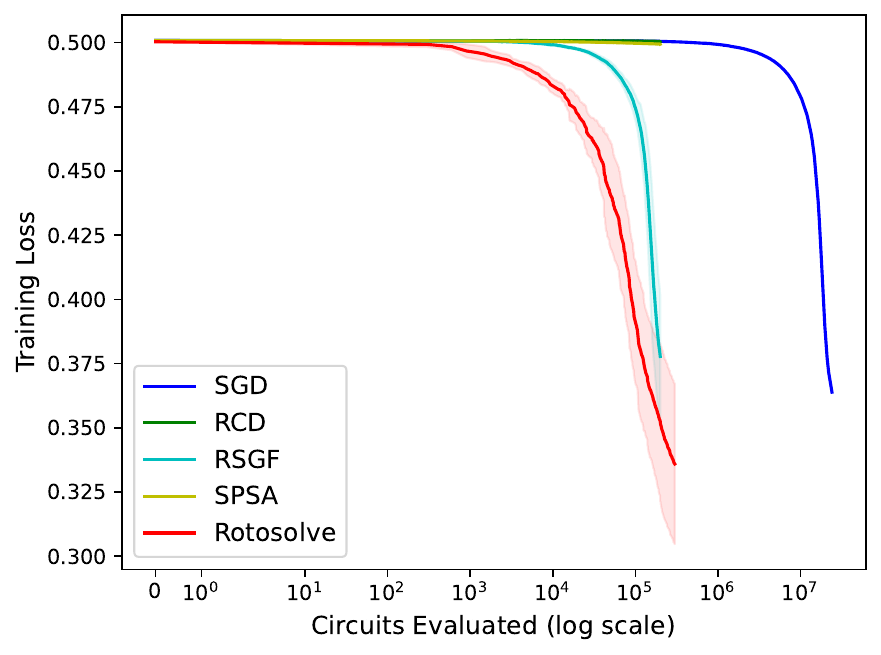}
        \caption{}
    \end{subfigure}
    \caption{Plot of training loss against (a) iterations, and (b) the number of circuit-executions (on a log-scale) for SGD, RCD, RSGF, SPSA, and \roto.}
    \label{fig:res}
    \vspace{-0.4cm}
\end{figure}


\section{Conclusions}\label{sec:con}
In this paper, we provided a rigorous analysis of the convergence of the \roto algorithm, proved its convergence, and derived explicit rates. Our results highlight that it is equally important to study the convergence guarantees of VQA-specific, structure-exploiting optimizers, as it is to analyze generic optimizers in the context of VQAs. We further stress that \roto implicitly leverages first and second derivatives in its updates, and as such, should not be viewed as a zeroth-order method.

As future work, it would be interesting to perform a more refined analysis and derive sharper convergence rates to better explain the superior empirical performance of \texttt{\textbf{Rotosolve}}. Another important direction is to study the convergence of its multiple variants. Finally, one may investigate the existence of pathological objective functions which are designed to adversarially affect the performance of \texttt{\textbf{Rotosolve}}. For example, in such a function, minimization along one direction may lead to maximization along another. We expect that this work will provide a foundation for the analysis of structure-exploiting optimization methods in a quantum setting.

\bibliographystyle{IEEEtran}
\bibliography{refs.bib}

@article{vqe,
	doi = {10.1038/ncomms5213},
  
	url = {https://doi.org/10.1038%2Fncomms5213},
  
	year = 2014,
	month = {jul},
  
	publisher = {Springer Science and Business Media {LLC}
},
  
	volume = {5},
  
	number = {1},
  
	author = {Alberto Peruzzo and Jarrod McClean and Peter Shadbolt and Man-Hong Yung and Xiao-Qi Zhou and Peter J. Love and Al{\'{a}}n Aspuru-Guzik and Jeremy L. O'Brien},
  
	title = {A variational eigenvalue solver on a photonic quantum processor},
  
	journal = {Nature Communications}
}

@article{spsa1,
  title={AN OVERVIEW OF THE SIMULTANEOUS PERTURBATION METHOD FOR EFFICIENT OPTIMIZATION},
  author={James C. Spall},
  journal={Johns Hopkins Apl Technical Digest},
  year={1998},
  volume={19},
  pages={482-492}
}

@article{psr,
	doi = {10.1103/physreva.99.032331},
  
	url = {https://doi.org/10.1103%2Fphysreva.99.032331},
  
	year = 2019,
	month = {mar},
  
	publisher = {American Physical Society ({APS})},
  
	volume = {99},
  
	number = {3},
  
	author = {Maria Schuld and Ville Bergholm and Christian Gogolin and Josh Izaac and Nathan Killoran},
  
	title = {Evaluating analytic gradients on quantum hardware},
  
	journal = {Physical Review A}
}

@article{sgd,
	doi = {10.22331/q-2020-08-31-314},
  
	url = {https://doi.org/10.22331%2Fq-2020-08-31-314},
  
	year = 2020,
	month = {aug},
  
	publisher = {Verein zur Forderung des Open Access Publizierens in den Quantenwissenschaften},
  
	volume = {4},
  
	pages = {314},
  
	author = {Ryan Sweke and Frederik Wilde and Johannes Meyer and Maria Schuld and Paul K. Faehrmann and Barth{\'{e}
}l{\'{e}}my Meynard-Piganeau and Jens Eisert},
  
	title = {Stochastic gradient descent for hybrid quantum-classical optimization},
  
	journal = {Quantum}
}

@article{rcd,
  title = {Random coordinate descent: A simple alternative for optimizing parameterized quantum circuits},
  author = {Ding, Zhiyan and Ko, Taehee and Yao, Jiahao and Lin, Lin and Li, Xiantao},
  journal = {Phys. Rev. Res.},
  volume = {6},
  issue = {3},
  pages = {033029},
  numpages = {25},
  year = {2024},
  month = {Jul},
  publisher = {American Physical Society},
  doi = {10.1103/PhysRevResearch.6.033029},
  url = {https://link.aps.org/doi/10.1103/PhysRevResearch.6.033029}
}

@ARTICLE{spsa2,
  author={Spall, J.C.},
  journal={IEEE Transactions on Aerospace and Electronic Systems}, 
  title={Implementation of the simultaneous perturbation algorithm for stochastic optimization}, 
  year={1998},
  volume={34},
  number={3},
  pages={817-823},
  doi={10.1109/7.705889}}

@article{vqa,
   title={Variational quantum algorithms},
   volume={3},
   ISSN={2522-5820},
   url={http://dx.doi.org/10.1038/s42254-021-00348-9},
   DOI={10.1038/s42254-021-00348-9},
   number={9},
   journal={Nature Reviews Physics},
   publisher={Springer Science and Business Media LLC},
   author={Cerezo, M. and Arrasmith, Andrew and Babbush, Ryan and Benjamin, Simon C. and Endo, Suguru and Fujii, Keisuke and McClean, Jarrod R. and Mitarai, Kosuke and Yuan, Xiao and Cincio, Lukasz and Coles, Patrick J.},
   year={2021},
   month=aug, pages={625–644} }

@article{gpsr,
   title={General parameter-shift rules for quantum gradients},
   volume={6},
   ISSN={2521-327X},
   url={http://dx.doi.org/10.22331/q-2022-03-30-677},
   DOI={10.22331/q-2022-03-30-677},
   journal={Quantum},
   publisher={Verein zur Forderung des Open Access Publizierens in den Quantenwissenschaften},
   author={Wierichs, David and Izaac, Josh and Wang, Cody and Lin, Cedric Yen-Yu},
   year={2022},
   month=mar, pages={677} }

@inproceedings{amira,
 author = {Abbas, Amira and King, Robbie and Huang, Hsin-Yuan and Huggins, William J. and Movassagh, Ramis and Gilboa, Dar and McClean, Jarrod},
 booktitle = {Advances in Neural Information Processing Systems},
 title = {On quantum backpropagation, information reuse, and cheating measurement collapse},
 url = {https://proceedings.neurips.cc/paper_files/paper/2023/file/8c3caae2f725c8e2a55ecd600563d172-Paper-Conference.pdf},
 year = {2023}
}

@misc{pl,
      title={Linear Convergence of Gradient and Proximal-Gradient Methods Under the Polyak-\L{}ojasiewicz Condition}, 
      author={Hamed Karimi and Julie Nutini and Mark Schmidt},
      year={2020},
      eprint={1608.04636},
      archivePrefix={arXiv},
      primaryClass={cs.LG}
}

@article{liu_noise,
  title = {Stochastic noise can be helpful for variational quantum algorithms},
  author = {Liu, Junyu and Wilde, Frederik and Mele, Antonio Anna and Jin, Xin and Jiang, Liang and Eisert, Jens},
  journal = {Phys. Rev. A},
  volume = {111},
  issue = {5},
  pages = {052441},
  numpages = {14},
  year = {2025},
  month = {May},
  publisher = {American Physical Society},
  doi = {10.1103/PhysRevA.111.052441},
  url = {https://link.aps.org/doi/10.1103/PhysRevA.111.052441}
}

@ARTICLE{Ghadimi_Lan,
  title     = "Stochastic first- and zeroth-order methods for nonconvex
               stochastic programming",
  author    = "Ghadimi, Saeed and Lan, Guanghui",
  journal   = "SIAM J. Optim.",
  publisher = "Society for Industrial \& Applied Mathematics (SIAM)",
  volume    =  23,
  number    =  4,
  pages     = "2341--2368",
  month     =  jan,
  year      =  2013,
  language  = "en"
}

@misc{sgd_handbook,
      title={Handbook of Convergence Theorems for (Stochastic) Gradient Methods}, 
      author={Guillaume Garrigos and Robert M. Gower},
      year={2024},
      eprint={2301.11235},
      archivePrefix={arXiv},
      primaryClass={math.OC},
      url={https://arxiv.org/abs/2301.11235}, 
}

@ARTICLE{Nesterov,
  title     = "Random gradient-free minimization of convex functions",
  author    = "Nesterov, Yurii and Spokoiny, Vladimir",
  journal   = "Found. Comut. Math.",
  publisher = "Springer Science and Business Media LLC",
  volume    =  17,
  number    =  2,
  pages     = "527--566",
  month     =  apr,
  year      =  2017,
  language  = "en"
}

@article{rotosolve,
  doi = {10.22331/q-2021-01-28-391},
  url = {https://doi.org/10.22331/q-2021-01-28-391},
  title = {Structure optimization for parameterized quantum circuits},
  author = {Ostaszewski, Mateusz and Grant, Edward and Benedetti, Marcello},
  journal = {{Quantum}},
  issn = {2521-327X},
  publisher = {{Verein zur F{\"{o}}rderung des Open Access Publizierens in den Quantenwissenschaften}},
  volume = {5},
  pages = {391},
  month = jan,
  year = {2021}
}

@article{ho_psr,
  title={Estimating the gradient and higher-order derivatives on quantum hardware},
  author={Mari, Andrea and Bromley, Thomas R and Killoran, Nathan},
  journal={Physical Review A},
  volume={103},
  number={1},
  pages={012405},
  year={2021},
  publisher={APS}
}

@article{smo,
   title={Sequential minimal optimization for quantum-classical hybrid algorithms},
   volume={2},
   ISSN={2643-1564},
   url={http://dx.doi.org/10.1103/PhysRevResearch.2.043158},
   DOI={10.1103/physrevresearch.2.043158},
   number={4},
   journal={Physical Review Research},
   publisher={American Physical Society (APS)},
   author={Nakanishi, Ken M. and Fujii, Keisuke and Todo, Synge},
   year={2020},
   month=Oct }

@article{icd,
   title={Interpolation-based coordinate descent method for parameterized quantum circuits},
   volume={9},
   ISSN={2399-3650},
   url={http://dx.doi.org/10.1038/s42005-025-02473-8},
   DOI={10.1038/s42005-025-02473-8},
   number={1},
   journal={Communications Physics},
   publisher={Springer Science and Business Media LLC},
   author={Lai, Zhijian and Hu, Jiang and Ko, Taehee and Wu, Jiayuan and An, Dong},
   year={2026},
   month=Jan }

@INPROCEEDINGS{ssd,
      title={Stochastic Shadow Descent: Training Parametrized Quantum Circuits with Shadows of Gradients}, 
      author={Sayantan Pramanik and M Girish Chandra},
      booktitle={To be presented at IEEE International Conference on Acoustics, Speech and Signal Processing (ICASSP), 2026}, 
      url={https://arxiv.org/abs/2511.12168}, 
}

@misc{
vqa_smooth,
title={VARIATIONAL QUANTUM ALGORITHMS ARE {L}IPSCHITZ SMOOTH},
author={Christopher Kverne and Mayur Akewar and Nicholas S. DiBrita and Yuqian Huo and Tirthak Patel and Janki Bhimani},
year={2026},
url={https://openreview.net/forum?id=beg6QFuff4}
}

@ARTICLE{excitation_solve,
  title     = "Fast gradient-free optimization of excitations in variational
               quantum eigensolvers",
  author    = "J{\"a}ger, Jonas and Kaldenbach, Thierry N and Haas, Max and
               Schultheis, Erik",
  journal   = "Commun. Phys.",
  publisher = "Springer Science and Business Media LLC",
  volume    =  8,
  number    =  1,
  pages     = "418",
  month     =  oct,
  year      =  2025,
  copyright = "https://creativecommons.org/licenses/by/4.0",
  language  = "en"
}

@misc{calculusonparameterizedquantumcircuits,
      title={Calculus on parameterized quantum circuits}, 
      author={Javier Gil Vidal and Dirk Oliver Theis},
      year={2018},
      eprint={1812.06323},
      archivePrefix={arXiv},
      primaryClass={quant-ph},
      url={https://arxiv.org/abs/1812.06323}, 
}

@misc{jacobidiagonalizationandersonacceleration,
      title={A Jacobi Diagonalization and Anderson Acceleration Algorithm For Variational Quantum Algorithm Parameter Optimization}, 
      author={Robert M. Parrish and Joseph T. Iosue and Asier Ozaeta and Peter L. McMahon},
      year={2019},
      eprint={1904.03206},
      archivePrefix={arXiv},
      primaryClass={quant-ph},
      url={https://arxiv.org/abs/1904.03206}, 
}

@INPROCEEDINGS{fraxis,
  author={Watanabe, Hiroshi C. and Raymond, Rudy and Ohnishi, Yu-Ya and Kaminishi, Eriko and Sugawara, Michihiko},
  booktitle={2021 IEEE International Conference on Quantum Computing and Engineering (QCE)}, 
  title={Optimizing Parameterized Quantum Circuits with Free-Axis Selection}, 
  year={2021},
  volume={},
  number={},
  pages={100-111},
  doi={10.1109/QCE52317.2021.00026}}

@article{fqs,
   title={Sequential optimal selections of single-qubit gates in parameterized quantum circuits},
   volume={9},
   ISSN={2058-9565},
   url={http://dx.doi.org/10.1088/2058-9565/ad4583},
   DOI={10.1088/2058-9565/ad4583},
   number={3},
   journal={Quantum Science and Technology},
   publisher={IOP Publishing},
   author={Wada, Kaito and Raymond, Rudy and Sato, Yuki and Watanabe, Hiroshi C},
   year={2024},
   month=May, pages={035030} }

@misc{fqs_controlled,
      title={Optimizing a parameterized controlled gate using Free Quaternion Selection}, 
      author={Hiroyoshi Kurogi and Katsuhiro Endo and Yuki Sato and Michihiko Sugawara and Kaito Wada and Kenji Sugisaki and Shu Kanno and Hiroshi C. Watanabe and Haruyuki Nakano},
      year={2025},
      eprint={2409.13547},
      archivePrefix={arXiv},
      primaryClass={quant-ph},
      url={https://arxiv.org/abs/2409.13547}, 
}

@INPROCEEDINGS{rotolasso,
  title           = "Optimizing parameters of quantum circuits with
                     sparsity-inducing coordinate descent",
  booktitle       = "Proceedings of the {Thirty-Fourth} International Joint
                     Conference on Artificial Intelligence",
  author          = "Raymond, Rudy and He, Zichang",
  publisher       = "International Joint Conferences on Artificial Intelligence
                     Organization",
  pages           = "6111--6119",
  month           =  sep,
  year            =  2025,
  address         = "California",
  conference      = "Thirty-Fourth International Joint Conference on Artificial
                     Intelligence \{IJCAI-25\}",
  location        = "Montreal, Canada"
}

@misc{improvingvariationalquantumcircuit,
      title={Improving Variational Quantum Circuit Optimization via Hybrid Algorithms and Random Axis Initialization}, 
      author={Joona V. Pankkonen and Lauri Ylinen and Matti Raasakka and Ilkka Tittonen},
      year={2025},
      eprint={2503.20728},
      archivePrefix={arXiv},
      primaryClass={quant-ph},
      url={https://arxiv.org/abs/2503.20728}, 
}

@ARTICLE{Nesterov_coordinate_descent,
  title     = "Efficiency of coordinate descent methods on huge-scale
               optimization problems",
  author    = "Nesterov, Yu",
  journal   = "SIAM J. Optim.",
  publisher = "Society for Industrial \& Applied Mathematics (SIAM)",
  volume    =  22,
  number    =  2,
  pages     = "341--362",
  month     =  jan,
  year      =  2012
}

@misc{Wright_coordinate_descent,
      title={Coordinate Descent Algorithms}, 
      author={Stephen J. Wright},
      year={2015},
      eprint={1502.04759},
      archivePrefix={arXiv},
      primaryClass={math.OC},
      url={https://arxiv.org/abs/1502.04759}, 
}

@inproceedings{finite_sum,
 author = {Arjevani, Yossi and Shamir, Ohad},
 booktitle = {Advances in Neural Information Processing Systems},
 editor = {D. Lee and M. Sugiyama and U. Luxburg and I. Guyon and R. Garnett},
 pages = {},
 publisher = {Curran Associates, Inc.},
 title = {Dimension-Free Iteration Complexity of Finite Sum Optimization Problems},
 url = {https://proceedings.neurips.cc/paper_files/paper/2016/file/299570476c6f0309545110c592b6a63b-Paper.pdf},
 volume = {29},
 year = {2016}
}

@misc{adam,
      title={Adam: A Method for Stochastic Optimization}, 
      author={Diederik P. Kingma and Jimmy Ba},
      year={2017},
      eprint={1412.6980},
      archivePrefix={arXiv},
      primaryClass={cs.LG},
      url={https://arxiv.org/abs/1412.6980}, 
}

\end{document}